\documentclass[11pt]{article}
\usepackage{amsmath,amssymb,amscd,amsthm,amsfonts}
\usepackage{graphicx,subcaption,enumerate}
\usepackage{hyperref}
\usepackage{dsfont}
\usepackage{color}
\usepackage{float}
\usepackage{multirow}
\usepackage{algorithm}
\usepackage[noend]{algpseudocode}
\usepackage[margin=1in]{geometry}
\title{Unfolding Orthotubes with a Dual Hamiltonian Path}
\author{Erik D. Demaine\thanks{Computer Science and Artificial Intelligence Laboratory, Massachusetts Institute of Technology, Cambridge, MA 02139, USA. \texttt{edemaine@mit.edu} }
\and Kritkorn Karntikoon\thanks{Department of Computer Science, Princeton University, Princeton, NJ 08544, USA. \texttt{kritkorn@princeton.edu}}}
\date{}

\newtheorem{theorem}{Theorem}
\newtheorem{lemma}[theorem]{Lemma}

\theoremstyle{definition}
\newtheorem{definition}{Definition}
\newtheorem{case}{Case}
\newtheorem{subcase}{Case}
\numberwithin{subcase}{case}

{\makeatletter
 \gdef\xxxmark{%
   \expandafter\ifx\csname @mpargs\endcsname\relax 
     \expandafter\ifx\csname @captype\endcsname\relax 
       \marginpar{xxx}
     \else
       xxx 
     \fi
   \else
     xxx 
   \fi}
 \gdef\xxx{\@ifnextchar[\xxx@lab\xxx@nolab}
 \long\gdef\xxx@lab[#1]#2{\textbf{[\xxxmark #2 ---{\sc #1}]}}
 \long\gdef\xxx@nolab#1{\textbf{[\xxxmark #1]}}
}

\let\epsilon=\varepsilon

\newcommand\qturn{\operatorname{qturn}}
\newcommand\start{\operatorname{start}}
\newcommand\hole{\operatorname{hole}}

\def\defn#1{\textbf{\textit{\boldmath #1}}}

\usepackage[usenames, dvipsnames]{xcolor}

\begin{document}

\maketitle

\begin{abstract}
An \defn{orthotube} consists of orthogonal boxes (e.g., unit cubes)
glued face-to-face to form a path.
In 1998, Biedl et al.\ showed that every orthotube has a \defn{grid unfolding}:
a cutting along edges of the boxes so that
the surface unfolds into a connected planar shape without overlap.
We give a new algorithmic grid unfolding of orthotubes
with the additional property that
the rectangular faces are attached in a single path ---
a Hamiltonian path on the rectangular faces of the orthotube surface.
\end{abstract}

\section{Introduction}\label{section-introduction}

Does every orthogonal polyhedron have a \defn{grid unfolding}, that is,
a cutting along edges of the induced grid (extending a plane through every
face of the polyhedron) such that the remaining surface unfolds into a
connected planar shape without overlap?
This question remains unsolved over 20 years after this type of unfolding
was introduced in 1998 \cite{demainez1998unfolding};
see \cite{o2008unfolding} for a survey and
\cite{genus2,DBLP:conf/cccg/DamianF18,damian2021unfolding}
for recent progress.
This problem is in some sense the orthogonal nonconvex version of the
older and more famous open problem of whether every convex polyhedron
has an edge unfolding (cutting only along edges of the polyhedron)
\cite{Demaine:2008:GFA:1478766}.
Many subclasses of orthogonal polyhedra have been successfully unfolded,
though sometimes cutting along more than just grid edges
\cite{Bern2003UnunfoldablePW,damian2004edge,damian2005unfolding,damian2007epsilon,o2008unfolding,damian2014unfolding,Liou2014OnEO,yjchang2015,genus2,kyho2019arbitrary,DBLP:conf/cccg/DamianF18,damian2021unfolding}. 


The first class of orthogonal polyhedra shown to have a grid unfolding is
\defn{orthotubes} \cite{demainez1998unfolding}, formed by gluing together
a sequence of orthogonal boxes where every pair of consecutive boxes in the
sequence share one face (and no other boxes share faces).
Roughly speaking, this unfolding consists of a monotone dual-path of
rectangular faces, with $O(1)$ rectangles attached above and below the path.

\begin{figure}
\centering
\subcaptionbox{orthotube}{\includegraphics[width=0.20\columnwidth]{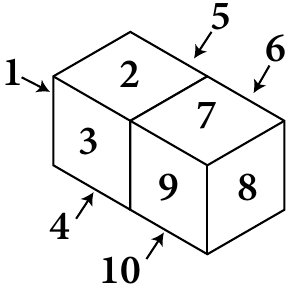}}\hfill
\subcaptionbox{face adjacency graph}{\includegraphics[width=0.25\columnwidth]{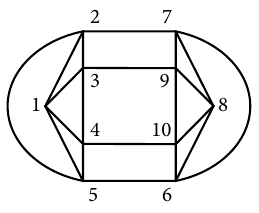}}\hfill
\subcaptionbox{\centering dual-Hamiltonian unfolding}{\includegraphics[width=0.22\columnwidth]{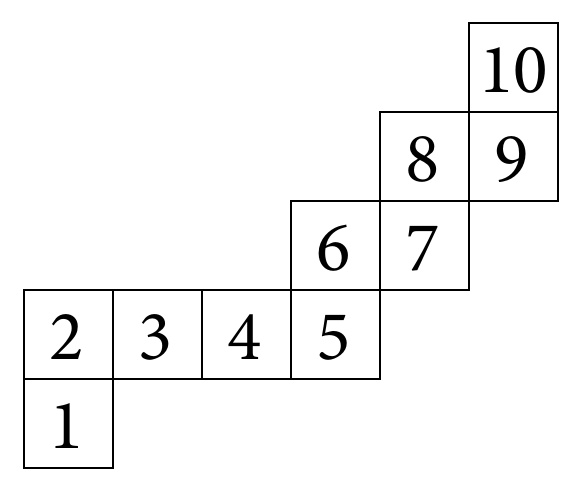}}\hfill
\subcaptionbox{\label{figure:example:dual}\centering dual Hamiltonian path}{\includegraphics[width=0.22\columnwidth]{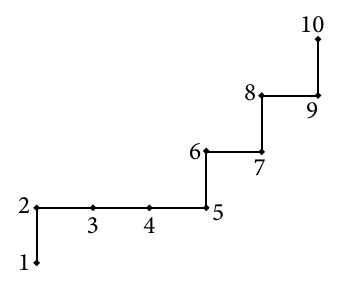}}
\caption{An example of an orthotube, its face adjacency graph, dual-Hamiltonian unfolding corresponding to chain code $RSSLRLRL$, and the corresponding dual Hamiltonian path.}
\label{figure:example}
\end{figure} 

In this paper, we show that orthotubes have a grid unfolding with a stronger
property we call \defn{dual-Hamiltonicity}, where the unfolded shape consists
of a single dual path of rectangular faces,
as shown in Figure~\ref{figure:example}.
More precisely, define the
\defn{face adjacency graph} to have a node for each rectangular face of a box,
and connect two nodes by an edge whenever the corresponding rectangular faces
share an edge.
Then the unfolding is given by keeping attached a chain of rectangular faces
corresponding to a Hamiltonian path in the face adjacency graph.
Implicitly, we take advantage of the fact that 4-connected planar graphs
(which includes face adjacency graphs) have Hamiltonian cycles
\cite{thomassen1983theorem, chiba1989hamiltonian}.
This fact has also been studied previously in the context of vertex unfolding \cite{garcia2018vertex} and grid unfolding as \textit{zipper unfolding} \cite{lubiwDDSS10, demaineDU13}.



The paper is organized as follows.
First, Section \ref{section-definition} defines useful tools
for expressing dual-Hamiltonian unfoldings.
Then Section~\ref{section-algorithm} describes and proves correct
our algorithm to find such an unfolding for a given orthotube.
Finally, in Section~\ref{section-conclusion}, we describe possible
future extensions to our result.

\section{Chain Codes for Dual Paths}\label{section-definition}{}

As mentioned in Section~\ref{section-introduction},
our unfolding keeps attached a chain of rectangular faces
corresponding to a Hamiltonian path in the face adjacency graph.
In this section, we introduce a tool for describing such
dual-Hamiltonian paths called ``chain codes''
(similar to \cite{damian2007epsilon,zachary2011NP,LEMUS20141721,LEMUS2015101}):

\begin{definition} \label{def:chain code}
A \defn{chain code} is an ordered sequence of the form $a_1 a_2 \cdots a_n$,
where $a_i \in \{L, R, S\}$ represents an (intrinsic) left turn, right turn,
or continuing straight to move from the $i$th face to the $(i+1)$st face. 

Given a starting face $f_0$ and initial intrinsic direction on the surface
(equivalently, the next face $f_1$ visited),
a chain code $a_1 a_2 \cdots a_n$ defines a \defn{corresponding path}
$f_0, f_1, f_2, \dots, f_{n+1}$ in the face adjacency graph;
see Figure~\ref{figure:turnLRS}.
\end{definition}

\begin{figure}[htbp]
\centering
\subcaptionbox{Turning left   ($L$)}{\includegraphics[scale=0.8]{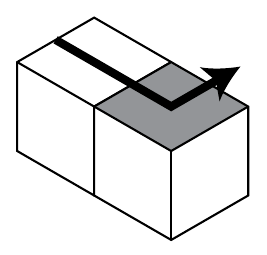}}\hfil
\subcaptionbox{Turning right  ($R$)}{\includegraphics[scale=0.8]{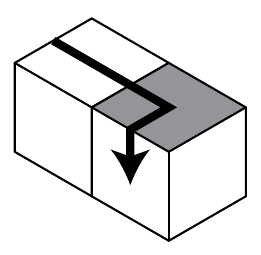}}\hfil
\subcaptionbox{Going straight ($S$)}{\includegraphics[scale=0.8]{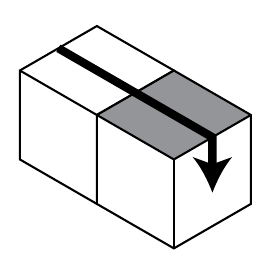}}
\caption{The three possible changes of direction for a dual path on the surface of an orthotube (forbidding U-turns). The turn occurs within the gray box, changing the entering direction from the previous box face to the exiting direction to next box face.}
\label{figure:turnLRS}
\end{figure}

If we unfold an orthotube to keep attached the path of faces
corresponding to a chain code,
then we can construct the 2D geometry of the unfolding
by following the intrinsic directions as follows:

\begin{definition}
The \defn{unfolding dual} of a chain code $a_1 a_2 \cdots a_n$
is the orthogonal equilateral path $p_0, p_1, \dots, p_n$
in the $xy$ plane that starts with the segment
from $p_0=(0,0)$ to $p_1=(0,1)$
and where the $i$th vertex $p_i$ turns left, turns right, or goes straight
according to $a_i \in \{L, R, S\}$.
The \defn{corresponding unfolding} has a unit square centered at
each vertex $p_i$, where the $i$th and $(i+1)$st squares are attached
along the $90^\circ$ rotation of the segment $p_i p_{i+1}$.
\end{definition}

For example, if the $i$th segment of the unfolding dual was from
$(x, y-1)$ to $(x,y)$, then the next vertex $i+1$ in the unfolding dual
is $(x-1, y)$ if $a_i = L$;
$(x+1, y)$ if $a_i = R$;
and $(x, y+1)$ if $a_i = S$.
Figure~\ref{figure:example:dual} shows another example
when the chain code is $RSSLRLRL$.
We can characterize nonoverlap as follows:

\begin{lemma} \label{lemma:unfolding no overlap}
  If the unfolding dual of a chain code has no overlapping vertices
  in the $xy$ plane, then the corresponding unfolding has no overlap.
\end{lemma}

\begin{proof}
  By the correspondence between points $p_i$ of the unfolding dual
  and unit squares in the corresponding unfolding,
\end{proof}

\label{invariant}
To guarantee that our unfolding does not overlap, we prove the invariant that
the unfolding dual proceeds monotonically in the $+y$ direction.
We can measure this property more easily using the following notion:

\begin{definition}
For each chain code $U = a_1 a_2 \cdots a_n$, the
\defn{cumulative quarter turning} $\qturn(U)$ is
$\sum_{i=1}^n \qturn(a_i)$
where $\qturn(R) = +1$, $\qturn(L) = -1$, and $\qturn(S) = 0$.
\end{definition}

Our desired invariant of the unfolding dual proceeding monotonically
in the $+y$ direction is thus equivalent to requiring that
$\qturn$ of any prefix of the chain code is in $\{-1, 0, +1\}$.

\begin{lemma} \label{lemma:monotone no overlap}
  If a chain code $U$ satisfies $\qturn(P) \in \{-1, 0, +1\}$
  for every prefix $P$ of~$U$,
  then its corresponding unfolding has no overlap.
\end{lemma}

\begin{proof}
  Divide $U$ into segments $S_1 S_2 \cdots S_k$
  where every prefix $P_j = S_1 S_2 \cdots S_j$ of segments
  has $\qturn(P_j) = 0$,
  and no other prefix $P$ of $U$ has $\qturn(P) = 0$.
  Assume by induction that the unfolding dual of $P_{j-1}$ has no repeated
  points, and that the last point is strictly above all other points.
  (In the base case, $j-1=0$ and $P_{j-1}$ is empty,
  so the hypothesis holds for the two points $p_0,p_1$.)
  If the first symbol of $P_j$ is $S$, then we immediately
  enter a new row and segment, so the inductive hypothesis holds.
  If the first symbol of $P_j$ is $R$, then $\qturn$ becomes $+1$,
  so the remaining symbols of $P_j$ must be $S$ followed by a final~$L$,
  at which point $\qturn$ becomes~$0$.
  The corresponding points in the unfolding dual form a rightward row (above all other points), and the next point starts a new row,
  establishing the inductive hypothesis in this case.
  If the first symbol of $P_j$ is $L$, then a symmetric argument holds.
  Therefore the unfolding dual has no repeated points, so by
  Lemma~\ref{lemma:unfolding no overlap}, the corresponding unfolding
  does not overlap.
\end{proof}

\section{Algorithm}\label{section-algorithm}

In this section, we present an algorithm to unfold an orthotube along a
dual Hamiltonian path, by constructing a chain code for that path
(Definition~\ref{def:chain code}).
The algorithm builds the chain code incrementally, progressively unfolding
each box in the orthotube, while maintaining the invariant
(mentioned in Section~\ref{invariant}) that the unfolding's chain code so far
has $\qturn \in \{-1, 0, +1\}$.
By Lemma~\ref{lemma:monotone no overlap},
this invariant implies that the unfolding does not overlap.
In fact, to simplify the induction, the algorithm creates a chain code
for the unfolding of possibly several boxes at a time
to maintain the stronger condition that
the intermediate chain codes have $\qturn = 0$.
We divide into cases based on each box's relative position
to the next one, two, and sometimes three or four boxes (if they exist).

More formally, suppose the given orthotube consists of boxes
$B_0, B_1, \dots, B_n$ in order.
For some $k \geq 0$, we construct a chain code $U_k$
whose corresponding unfolding is a nonoverlapping unfolding of
the suborthotube $B_0, B_1, \dots, B_k$.
This chain code is thus a Hamiltonian path on the faces of boxes
$B_0, B_1, \dots, B_k$ that are actually faces of the orthotube ---
that is, excluding the shared face \defn{$\hole(i)$} between $B_i$ and $B_{i+1}$
for $0 \leq i < n$ (see Figure \ref{figure:defn:hole}).
Note in particular that the suborthotube $B_0, B_1, \dots, B_k$ is not
a normal orthotube (for $k < n$) because $B_k$ is missing one face, $\hole(k)$.
In the chain code $U_k$ we include a turn code for the last face of $B_k$
visited (for $k < n$), and require that this turn would bring the path
into a face of $B_{k+1}$; we call this requirement \defn{continuability}.
We construct $U_k$ inductively, using a previous $U_{k-x}$ with
$\qturn(U_{k-x})=0$ to construct $U_k$ with $\qturn(U_k) = 0$.
In the final step, we construct $U_n$ (which may not have $\qturn(U_n) = 0$),
which is the chain code for an unfolding of the entire orthotube.

\begin{figure}[H]
  \centering
  \begin{subfigure}[t]{0.45\linewidth}
    \centering
    \includegraphics[scale=0.8]{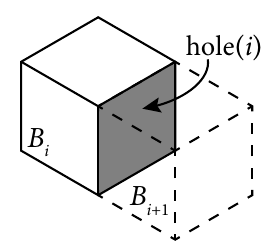}
    \caption{\label{figure:defn:hole}Face $\hole(i)$ is the shared face of $B_i$ and $B_{i+1}$, which is not a face of the orthotube.}
  \end{subfigure}\hfil
  \begin{subfigure}[t]{0.45\linewidth}
    \centering
    {\includegraphics[scale=0.8]{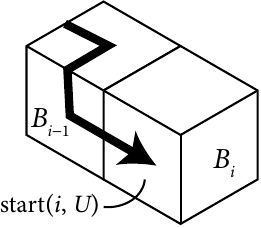}}
    \caption{\label{figure:defn:start}Face $\start(i, U)$ is the first face of $B_i$ encountered by the unfolding chain code $U$ (drawn as an arrow path).}
  \end{subfigure}

  \caption{Definitions of the faces $\hole(i)$ and $\start(i, U)$.}
  \label{figure:defn}
\end{figure} 

For the base case, we construct either $U_0 = LSSR$ or $U_0 = RSSL$;
see Figure~\ref{figure:base_case}.
In either case, $\qturn(U_0) = 0$.
By starting this chain code at the face of $B_0$ opposite the shared face
$\hole(0)$ with $B_1$,
it successfully visits all faces of $B_0$ except the shared face,
and the last turn code attempts to enter a face of~$B_1$.
Even fixing the initial direction for the chain code,
at least one of the two choices for $U_0$ will be continuable,
successfully entering a face of $B_1$ and not $B_2$.

\begin{figure}[H]
  \centering
  \begin{subfigure}[t]{0.45\linewidth}
    \centering
    \includegraphics[scale=0.8]{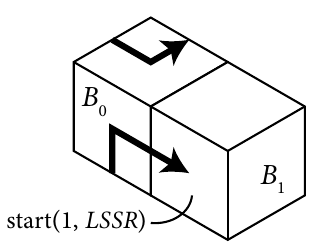}
    \caption{$LSSR$}
  \end{subfigure}\hfill
  \begin{subfigure}[t]{0.45\linewidth}
    \centering
    \includegraphics[scale=0.8]{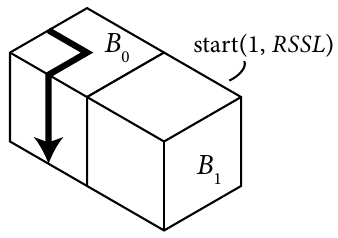}
    \caption{$RSSL$}
  \end{subfigure}

  \caption{Two possible ways to construct $U_0$ starting at the face of $B_0$ opposite to $\hole(0)$.}
  \label{figure:base_case}
\end{figure} 

For the inductive step, we are given a continuable chain code $U_{i-1}$
with $\qturn(U_{i-1}) = 0$.
By continuability, the last turn code of $U_{i-1}$ enters a face of $B_i$,
which we call \defn{$\start(i, U_{i-1})$}; see Figure~\ref{figure:defn:start}.
Now we unfold $B_i$ based on the relative position of the next box $B_{i+1}$,
in order to guarantee that $U_i$ is continuable into $B_{i+1}$.
There are four cases for the relative position,
which we denote \defn{$N(i,U_{i-1})$}, as follows;
refer to Figure~\ref{figure:pattern}.
%
\begin{itemize}
    \item $N(i, U) = S$ if $\start(i, U)$ can be directed to $\hole(i)$ by going straight, as in Figure \ref{figure:pattern:straight}.
    \item $N(i, U) = O$ if $\start(i, U)$ is on the opposite side of $\hole(i)$, as in Figure \ref{figure:pattern:opposite}.
    \item $N(i, U) = L$ if $\start(i, U)$ can be directed to $\hole(i)$ by turning left, as in Figure \ref{figure:pattern:left}.
    \item $N(i, U) = R$ if $\start(i, U)$ can be directed to $\hole(i)$ by turning right, as in Figure \ref{figure:pattern:right}.
\end{itemize}

\begin{figure}[H]
\centering
\subcaptionbox{\label{figure:pattern:straight}Straight}{\includegraphics[scale=0.6]{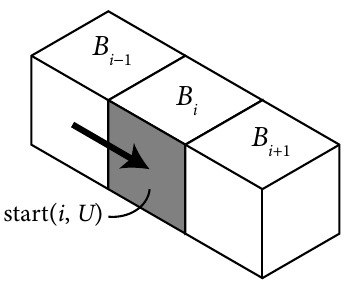}}\hfill
\subcaptionbox{\label{figure:pattern:opposite}Opposite}{\includegraphics[scale=0.6]{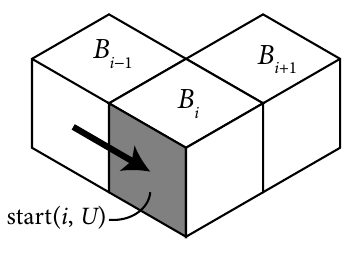}}\hfill
\subcaptionbox{\label{figure:pattern:left}Left}{\includegraphics[scale=0.6]{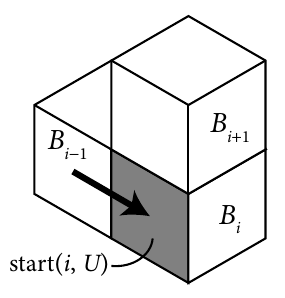}}\hfill
\subcaptionbox{\label{figure:pattern:right}Right}{\includegraphics[scale=0.6]{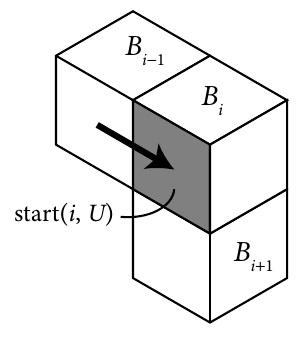}}

\caption{Four possible configurations of the next box $B_{i+1}$ (equivalently, $\hole(i)$) relative to box $B_i$ and face $\start(i, U)$. Each configuration gives a different value for $N(i, U)$.}
\label{figure:pattern}
\end{figure} 

The chain code needed to unfold box $B_i$ in each case is different.
For example, if $N(i, U) = S$ as shown in Figure \ref{figure:pattern:straight},
then we might unfold by $RSSL$ or $LSSR$;
while, if $N(i, U) = O$ as shown in Figure \ref{figure:pattern:opposite},
then we might unfold by $RLSR$ or $LRSL$ instead.  
But only some of these chain codes may be continuable, depending on
the relative configuration $N(i+2,U)$ of box $B_{i+2}$.

%

In the remainder of this section, we provide a chain code to append to
$U_{i-1}$ that has $\qturn = 0$,
determined by $N(i, U_{i-1})$ and its continuability. 
In some cases, such as when $N(i, U_{i-1}) = S$, we can find a continuable chain code to unfold the box $B_i$, extend the code $U_{i-1}$ to $U_i$ with $\qturn(U_i) = 0$, and continue to the induction step $i+1$. However, in most cases, the continuable chain code involves unfolding multiple boxes at a time in order to restore $\qturn$ to $0$, resulting in skipped step(s).
We also need to ensure that the provided chain code visits every face and the $\qturn$ of any prefix of the chain code is in $\{-1, 0, +1\}$. These facts are easy to check and will be omitted here.

\begin{case}\label{case-S} $N(i, U_{i-1}) = S$.

There are two ways to unfold $B_i$, $LSSR$ and $RSSL$;
see Figure~\ref{figure:caseS}.
Because both chain codes have $\qturn = 0$, then no matter which of these chain codes we append, we still have $\qturn(U_i) = 0$. Because $LSSR$ and $RSSL$ lead to different faces, at least one of them makes $U_i$ continuable.
\begin{figure}[h]
  \centering
  \begin{subfigure}[t]{0.45\linewidth}
    \centering
    \includegraphics[scale=0.8]{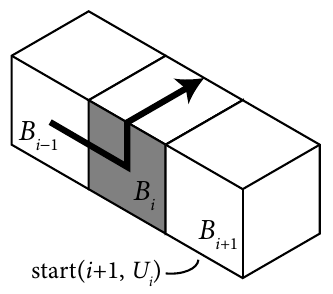}
    \caption{$LSSR$}
  \end{subfigure}\hfil
  \begin{subfigure}[t]{0.45\linewidth}
    \centering
    \includegraphics[scale=0.8]{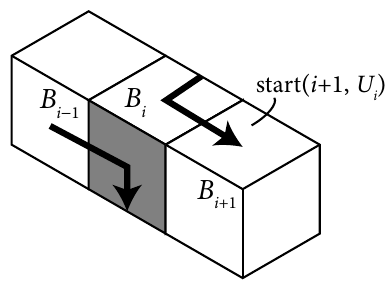}
    \caption{$RSSL$}
  \end{subfigure}

  \caption{Two possible ways to unfold $B_i$ when $N(i, U) = S$. The gray face is $\start(i, U_{i-1})$.}
  \label{figure:caseS}
\end{figure} 
\end{case}

\begin{case}\label{case-L} $N(i, U_{i-1}) = L$.

We consider two ways to unfold $B_i$, $RLRL$ and $RSLR$;
see Figure \ref{figure:caseL:a} and \ref{figure:caseL:b}.
We prefer adding $RLRL$ because its $\qturn$ is $0$.
Thus, if $U_{i-1} + RLRL$ is continuable,
then we assign $U_i = U_{i-1} + RLRL$, and this satisfies our invariants.%
\footnote{We use $+$ to denote concatenation of chain codes.}

In the subcase when $U_{i-1} + RLRL$ is not continuable,
we need to unfold $B_i$ with $RSLR$, i.e., set $U_i = U_{i-1} + RSLR$.
Thus $\qturn(U_i) = 1$, so we need to unfold the next box
in order to restore $\qturn$ to $0$.

\begin{figure}[H]
\centering
\subcaptionbox
  {\label{figure:caseL:a}$RLRL$}
  {\includegraphics[scale=0.7]{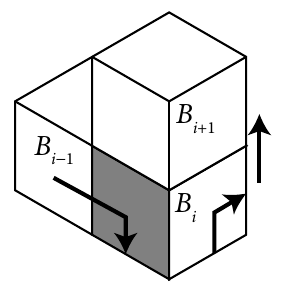}}\hfil
\subcaptionbox
  {\label{figure:caseL:b}$RSLR$}
  {\includegraphics[scale=0.7]{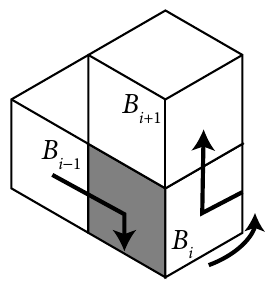}}\hfil
\subcaptionbox
  {\label{figure:caseL:c}The position of $B_{i+2}$ (red box) when pattern $RLRL$ is not continuable.}
  {\includegraphics[scale=0.7]{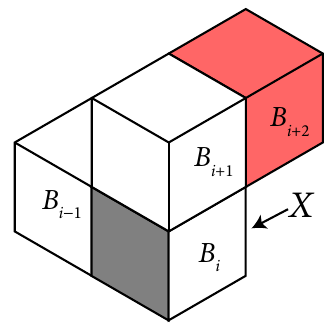}}

\caption{Two possible ways to unfold $B_i$ when $N(i, U) = L$, and an analysis of the latter case.}
\label{figure:caseL}
\end{figure} 

Because $U_{i-1} + RLRL$ is not continuable, the face $X$ of box $B_i$
opposite $\start(i,U_{i-1})$ (as shown in Figure \ref{figure:caseL:c})
must not be adjacent to box $B_{i+1}$, so it must be adjacent to box $B_{i+2}$.
Thus, after unfolding $B_i$ with $RSLR$, we have $N(i+1, U_i) = R$.
We consider unfolding $B_{i+1}$ with either $LRLR$ or $LSRL$
(the reflections of the unfoldings we considered for $B_i$,
to keep $\qturn \in \{-1,0,+1\}$),
but now we prefer $LSRL$ because it restores $\qturn(U_{i+1})$ to $0$. 

If $U_i + LSRL$ is continuable, then we assign $U_{i+1} = U_i + LSRL$
and satisfy the invariants.
Otherwise, we set $U_{i+1} = U_i + LRLR$ and need to continue unfolding,
as $\qturn(U_{i+1}) = 1$ still.
The following lemma guarantees that we remain in the same relative
configuration in this subsubcase:

\begin{lemma}\label{lemma-D}
If we have an integer $j$ and chain code $U$ such that $N(j, U) = R$ and
$U + LSRL$ is not continuable, then $N(j+1, U + LRLR) = R$.
\end{lemma}
\begin{proof}
Refer to Figure~\ref{figure:lemmaD:a}, which in particular shows $B_j$ and
$B_{j+1}$ with $N(j,U) = R$.
Given that $U + LSRL$ is not continuable, following $LSRL$ in $B_j$
does not lead us to the box $B_{j+1}$, so that face of $B_{j+1}$ must be
$\hole(j+1)$.  Focusing on boxes $B_j$ to $B_{j+2}$,
as in Figure~\ref{figure:lemmaD:b}, we see that $N(j+1,U+LRLR) = R$ as desired.
\end{proof}

\begin{figure}[h]
  \centering
  \begin{subfigure}[t]{0.30\linewidth}
    \centering
    \includegraphics[scale=0.7]{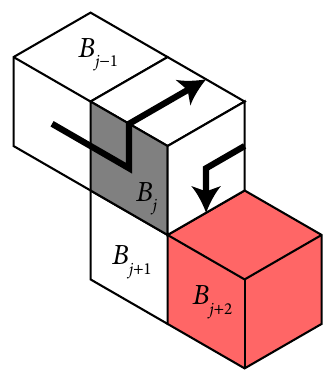}
    \caption{\label{figure:lemmaD:nc} A configuration from $B_{j-1}$ to $B_{j+2}$ where $N(j, U) = R$ and $U + LSRL$ is not continuable.}
  \end{subfigure}\hfill
  \begin{subfigure}[t]{0.30\linewidth}
    \centering
    \includegraphics[scale=0.7]{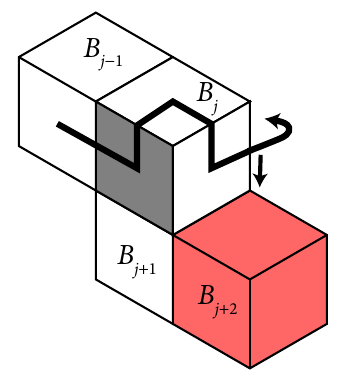}
    \caption{\label{figure:lemmaD:a} $U + LRLR$ is continuable.}
  \end{subfigure}\hfill
  \begin{subfigure}[t]{0.30\linewidth}
    \centering
    \includegraphics[scale=0.7]{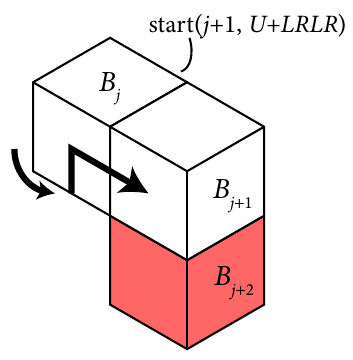}
    \caption{\label{figure:lemmaD:b}The same configuration, focusing on $B_j$ to $B_{j+2}$.}
  \end{subfigure}

  \caption{Illustration of the proof of Lemma~\ref{lemma-D}.}
  \label{figure:lemmaD}
\end{figure}

Because we remain in the same relative configuration by Lemma~\ref{lemma-D},
we can repeat the $LRLR$ unfolding until we reach the first positive integer
$k$ such that $U_{i-1+k} + LSRL$ is continuable.
Then we set
%
%
\[U_{i+k} = U_{i-1} + RSLR + (k-1) \times LRLR + LSRL\]
(where $\times$ denotes repetition of a chain code).
Thus $\qturn(U_{i+k}) = 0$ as desired, and we satisfy the invariants.

\end{case}

\begin{case}\label{case-R} $N(i, U_{i-1}) = R$.

We generate a chain code in the same way as in Case~\ref{case-L}
where $N(i, U_{i-1}) = L$, but swapping the roles of $L$ and $R$
in both $U$ and $N$ (effectively reflecting all diagrams).

\end{case}

\begin{case}\label{case-O} $N(i, U_{i-1}) = O$.

There are two ways to unfold $B_i$, $RLSR$ and $LRSL$,
with respective $\qturn$ values of $+1$ and $-1$.
Hence, again, we need to unfold the next box in order to
restore $\qturn$ to~$0$.
We do so exhaustively for all possible relative positions of $B_{i+2}$,
as shown in Figure~\ref{figure:caseO}.

\begin{figure}[H]
\centering
\subcaptionbox{\label{figure:caseO:a}}{\includegraphics[scale=0.5]{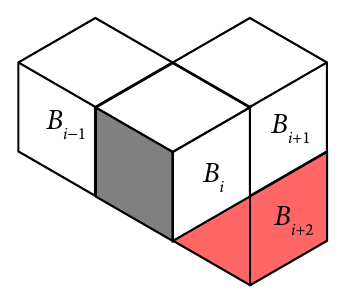}}\hfill
\subcaptionbox{\label{figure:caseO:b}}{\includegraphics[scale=0.5]{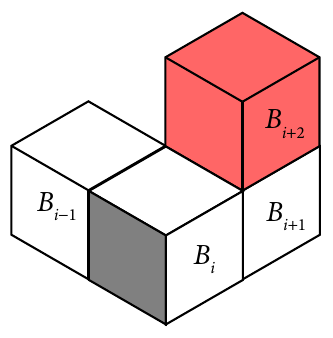}}\hfill
\subcaptionbox{\label{figure:caseO:c}}{\includegraphics[scale=0.5]{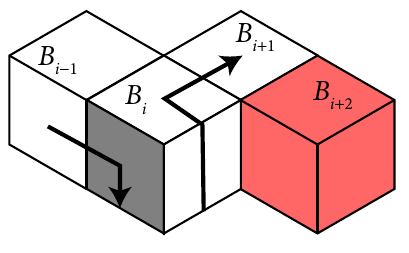}}\hfill
\subcaptionbox{\label{figure:caseO:d}}{\includegraphics[scale=0.5]{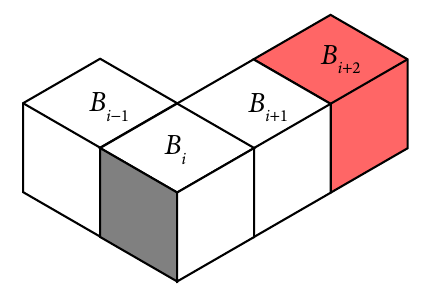}}

\caption{Four possible positions of $B_{i+2}$ (red box) when $N(i, U) = O$.}
\label{figure:caseO}
\end{figure} 

\begin{subcase}\label{caseO:a}
  $U_{i-1} + LRSL$ is not continuable,
  as shown in Figure \ref{figure:caseO:a}.

In this subcase, we need to further consider the position of $B_{i+1}$
(if it exists). 
Observe that $U_{i-1} + RLSR$ is continuable.
If $U_{i-1} + RLSR LRSL$ is continuable, then we assign $U_{i+1} = U_{i-1} + RLSR LRSL$, which restores $\qturn$ to $0$.

On the other hand, if $U_{i-1} + RLSR LRSL$ is not continuable, then
the position of $B_{i+3}$ is as shown in Figure~\ref{figure:caseO1_prob}.
We then assign $U_{i+2} = U_{i-1} + LRSL RLRS RLSS$ which restores $\qturn$ to $0$.

\begin{figure}[h]
  \centering
  \begin{subfigure}[t]{0.45\linewidth}
    \centering
    \includegraphics[scale=0.8]{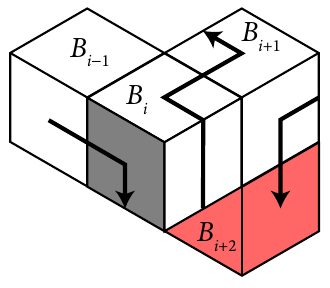}
    \caption{$RLSRLRSL$}
  \end{subfigure}\hfil
  \begin{subfigure}[t]{0.45\linewidth}
    \centering
    \includegraphics[scale=0.8]{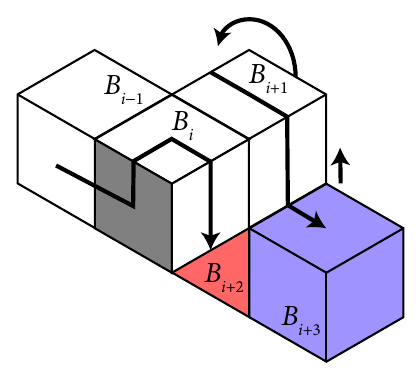}
    \caption{The extension of the configuration in Figure~\ref{figure:caseO:a} when $U_{i-1} + RLSR LRSL$ is not continuable. Box $B_{i+2}$ is red and $B_{i+3}$ is blue.}
  \end{subfigure}

  \caption{Two ways to generate the chain code for the configuration in Figure~\ref{figure:caseO:a}.}
  \label{figure:caseO1_prob}
\end{figure} 



\end{subcase}

\begin{subcase}\label{caseO:b}
  $U_{i-1} + RLSR$ is not continuable,
  as shown in Figure \ref{figure:caseO:b}.

  We generate a chain code as in the previous subcase
  by swapping the roles of $L$ and $R$ (effectively reflecting all diagrams).
\end{subcase}

\begin{subcase}\label{caseO:c}
  $B_{i+2}$ is ``coplanar'' with $B_{i-1}, B_i, B_{i+1}$,
  as shown in Figure \ref{figure:caseO:c}.

Because $U_{i-1} + RLSR$ is continuable,
we can assign $U_i = U_{i-1} + RLSR$.
Then we have $N(i+1, U_i) = R$ and $\qturn(U_i) = +1$.
If $U_i + LSRL$ is continuable, then we assign
\[U_{i+1} = U_i + LSRL = U_{i-1} + RLSR LSRL,\]
which restores $\qturn$ to~$0$.

On the other hand, if $U_i + LSRL$ is not continuable,
then we can use Lemma~\ref{lemma-D} to repeatedly unfold boxes with $LRLR$
until we reach a box where $LSRL$ is continuable,
similar to Case~\ref{case-L}.
Thus with the same method we obtain a positive integer $k$
and chain code $U_{i+k}$ with $\qturn(U_{i+k}) = 0$.

\end{subcase}

\begin{subcase}\label{caseO:d}
  $B_{i+2}$ and $B_i$ are on opposite sides of $B_{i+1}$,
  as shown in Figure~\ref{figure:caseO:d}.

In this subcase, we exhaustively examine the position of $B_{i+3}$ and provide the extended chain code accordingly. There are five possible relative positions for $B_{i+3}$ as follows;
refer to Figure~\ref{figure:caseO4}.

\begin{figure}[H]
\centering
\subcaptionbox{\label{figure:caseO4:a}}{\includegraphics[scale=0.38]{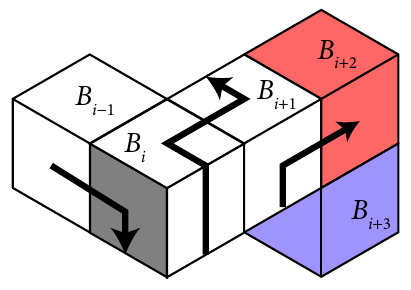}}\hfill
\subcaptionbox{\label{figure:caseO4:b}}{\includegraphics[scale=0.38]{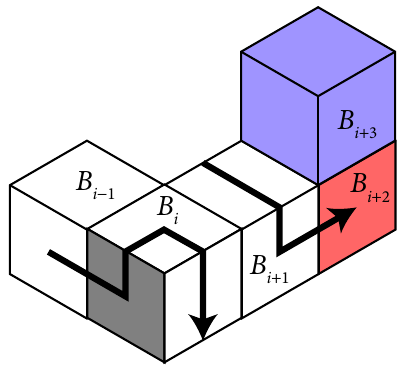}}\hfill
\subcaptionbox{\label{figure:caseO4:c}}{\includegraphics[scale=0.38]{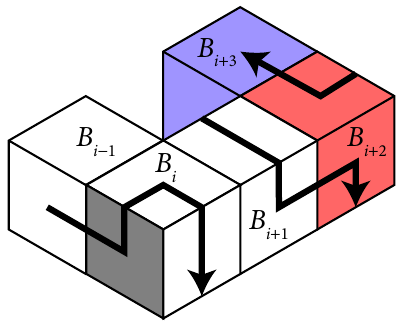}}\hfill
\subcaptionbox{\label{figure:caseO4:d}}{\includegraphics[scale=0.38]{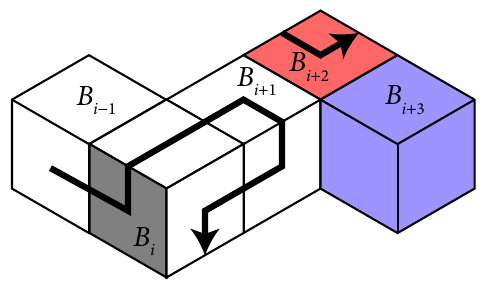}}\hfill
\subcaptionbox{\label{figure:caseO4:5}}{\includegraphics[scale=0.38]{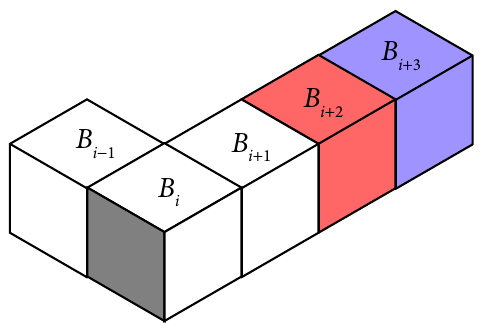}}

\caption{Five possible extensions of configuration shown in Figure \ref{figure:caseO:d}. Box $B_{i+2}$ is red and $B_{i+3}$ is blue.}
\label{figure:caseO4}
\end{figure}
\begin{enumerate}[(a)]
    \item Shown in Figure \ref{figure:caseO4:a}. We assign $U_{i+1} = U_{i-1} + RLSRLSSR$. Hence, we have $\qturn(U_{i+1}) = +1$ and $N(i+2, U_{i+1}) = R$. We can follow the same process as in Case \ref{caseO:c} to get a chain code for some $U_{i+k}$ with $\qturn(U_{i+k}) = 0$.
    \item Shown in Figure \ref{figure:caseO4:b}. We use the same method as the previous case but swap $L$ and $R$.
    \item Shown in Figure \ref{figure:caseO4:c}. If $U_{i-1} + LRSL RSSL RLSR$ is continuable, we assign that to $U_{i+2}$. Otherwise, we have the configuration shown in Figure \ref{figure:caseO4_3_prob}, and we assign $U_{i+2} = U_{i-1} + RLSR LSSR LRSL$ instead. 

    \begin{figure}[H]
    \centering
    \includegraphics[scale=0.7]{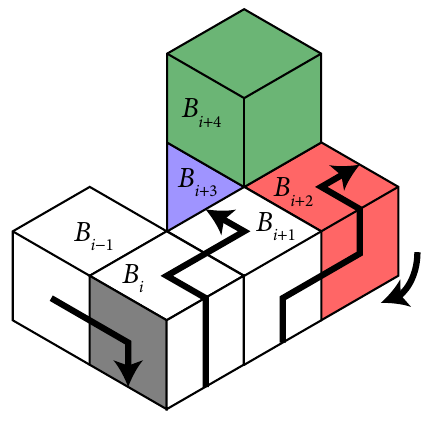}
    \caption{The extension of configuration in Figure \ref{figure:caseO4:c} when $U_{i-1} + LRSLRSSLRLSR$ is not continuable. Box $B_{i+4}$ is green.}
    \label{figure:caseO4_3_prob}
    \end{figure}

    \item Shown in Figure \ref{figure:caseO4:d}. If $U_{i-1} + LSRR LLRL RLSR$ is continuable, then we assign that to $U_{i+2}$. Otherwise, reflecting the extension, $U_{i-1} + RSLL RRLR LRSL$ must be continuable, and we assign that to $U_{i+2}$. In either case, we satisfy the invariants.
    \item In the last case, we have four boxes $B_i, B_{i+1}, B_{i+2}, B_{i+3}$
      in a straight line; refer to Figure~\ref{figure:caseO4:5}.
      
      Observe that as long as the boxes remain in a straight line, unfolding each box will preserve the $\qturn$ value because the only possible ways to unfold are $LSSR$ and $RSSL$. Hence, to restore $\qturn$ to $0$, we need to skip until we find the first box which is not lined up. 
      Let $k \geq 4$ be the smallest integer such that $B_{i+k}$ is not lined up with $B_i, B_{i+1}, B_{i+2}$.
      
      The idea is to unfold $B_i, B_{i+1},\dots$ until $B_{i+k-3}$ or $B_{i+k-2}$, and then apply the same methods used previously in Case~\ref{case-O} to generate chain code for the remaining boxes. The essence of why this works is that the chain code generated in Case~\ref{case-O} (except the forth subcase of Case~\ref{caseO:d}) always fully unfolds box $B_i$ before start unfolding $B_{i+1}$. Hence, if we can simulate the same $\qturn(U_i)$ and $\start(i+1, U_i)$, then we can treat $B_{i+k}$ as $B_{i+2}$ in Cases~\ref{caseO:a}--\ref{caseO:c} or as $B_{i+3}$ in Case~\ref{caseO:d} and follow the same method. 
      
      We claim that, for any $j > 0$, it is possible to unfold boxes $B_i$ to $B_{i+2j}$ in such a way that $\qturn(U_{i+2j})$ can be either $+1$ or $-1$ and $\start(i+2j+1, U_{i+2j})$ can be either top or bottom of the figure when looking from the perspective shown in Figure~\ref{figure:caseO4_5}.
      For $j=1$, the claim is true from  $\{U_{i-1} + RLSRLSSRLSSR, U_{i-1} + RSLLRRLRLSSR, U_{i-1} + LRSLRSSLRSSL, U_{i-1} + LSRRLLRLRSSL\}$, as shown in Figure~\ref{figure:caseO4_5}.
      To extend the result to any $j$, we need to unfold every two consecutive boxes with $RSSLRSSL$ if $\qturn(U_{i+1}) = -1$ and $LSSRLSSR$ if $\qturn(U_{i+1}) = +1$. This proves the claim.
      
      Next, we provide an algorithm to generate the chain code to unfold $B_i, \dots, B_{i+k}$.
      
      For even $k$, we first find the Case~\ref{caseO:a}--\ref{caseO:c} where their $B_i, B_{i+1}$ and $B_{i+2}$ are rearranged in the same way as boxes $B_{i+k-2}, B_{i+k-1}$, and $B_{i+k}$. (We also consider the upside-down version of Case~\ref{caseO:c}, so that it covers the case where $B_{i+k}$ is ``coplanar'' with $B_{i+k-1}, B_{i+k-2}$ and $B_{i+k-3}$ but has a different configuration from Case~\ref{caseO:c}.)
      Then, we generate the chain code according to the claim so we get the desired $\qturn(U_{i+k-2})$ and $\start(i+k-1, U_{i+k-2})$.
      
      For odd $k$, we follow a similar process as the even case, but with $B_{i+k-3}$ instead of $B_{i+k-2}$, and follow the same method for first three subcases of Case~\ref{caseO:d}. For the fourth case where the configuration is aligned with Figure~\ref{figure:caseO4:d}, we use the flipped chain code of Figure~\ref{figure:caseO4:c} instead because it fully unfolds $B_i$ first.

    \begin{figure}[H]
        \centering
        \subcaptionbox{}{\includegraphics[scale=0.4]{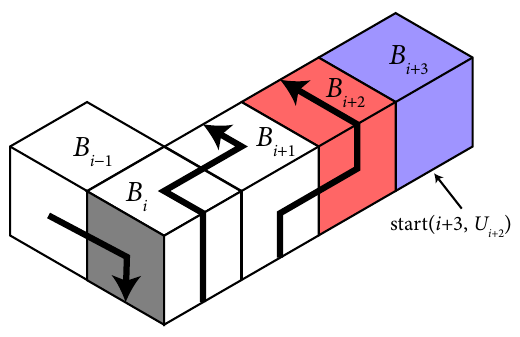}}\hfill
        \subcaptionbox{}{\includegraphics[scale=0.4]{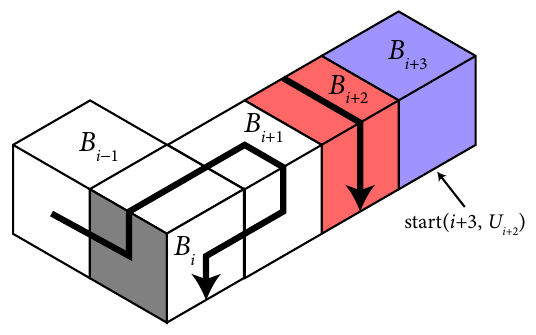}}\hfill
        \subcaptionbox{}{\includegraphics[scale=0.4]{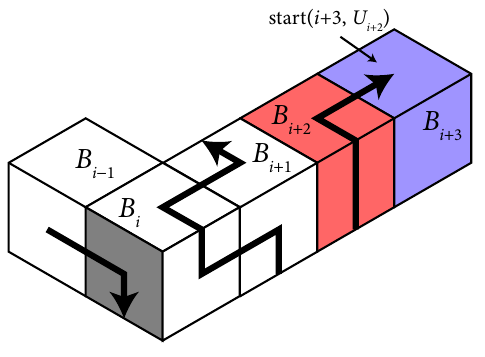}}\hfill
        \subcaptionbox{}{\includegraphics[scale=0.4]{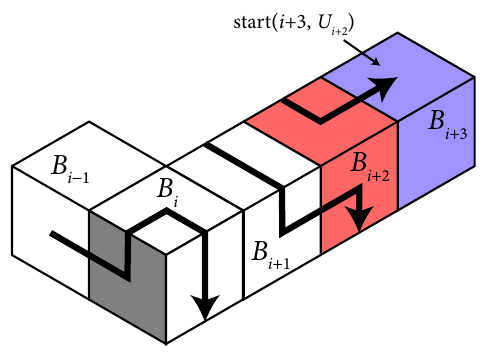}}
        \caption{Four ways to unfold $B_i$ to $B_{i+2}$ so we can have desired $\start(i+3, U_{i+2})$ and $\qturn(U_{i+2})$.}
        \label{figure:caseO4_5}
        \end{figure}
\end{enumerate}

\end{subcase}

\end{case}

Combining these cases, we have provided a chain code for every possible position of $B_{i+1}$ (and sometimes $B_{i+2}, B_{i+3}, \dots$). We ensure that the invariant for the induction is true: the $\qturn$ value of the chain code is always restored to $0$ (except when we reach $B_n$ first). 

To construct the final unfolding $U_n$, we can temporarily create a new box $B_{n+1}$ attached to $B_n$. By induction, we can find a chain code $U_n$ to unfold $B_0, B_1, \dots, B_n$ and end at $\start(n+1, U_n)$. Then when we remove $B_{n+1}$, the chain code $U_n$ will end at $\hole(n)$ instead, which results in an unfolding of the original orthotube. Because the $\qturn$ of any prefix of the chain code is in $\{-1, 0, +1\}$, by Lemma~\ref{lemma:monotone no overlap}, the corresponding unfolding of $U_n$ does not overlap (and is dual-Hamiltonian) as desired.

\section{Conclusion}\label{section-conclusion}
In this paper, we proposed a new algorithm to unfold orthotubes such that the unfolding path is a Hamiltonian path of the orthotube's face adjacency graph. In other words, we can unfold an orthotube by traveling through all faces on the orthotube's surface without having to visit the same face twice.

An intriguing harder goal is to find a Hamiltonian \emph{cycle} through the
face adjacency graph, such that breaking that cycle into a path results in
an unfolding without overlap.  This would require the unfolding to effectively
traverse the length of the orthotube twice (down and back up).
Potentially, such an approach could be more amenable for extension to
unfolding \defn{orthotrees} (boxes glued to form a tree instead of a path),
by recursing on subtrees and combining the cycles together.

\bibliographystyle{alpha}
\bibliography{references} 

\end{document}